\documentclass[journal, 10pt]{IEEEtran}

\usepackage{xcolor}
\usepackage{subfigure}	
\usepackage{graphicx,tipa}
\usepackage{amsmath,amssymb,amsthm}
\usepackage[lined,algonl,algoruled]{algorithm2e} 
\usepackage{algorithmic}
\usepackage{placeins}

\newtheorem{theorem}{Theorem}

\newtheorem{corollary}{Corollary}

\newtheorem{definition}{Definition}
\newcommand{\remove}[1]{}

\begin{document}

\title{Plane Sweep Algorithms for Data Collection in Wireless Sensor Network using Mobile Sink}

\author{
Dinesh Dash
\thanks{Dept. of CSE, National Institute of Technology Patna, India, e-mail: dd@nitp.ac.in }

\thanks{Manuscript received XXX, XX, 2018; revised XXX, XX, 2019.}}

\markboth{IEEE Transactions on Vehicular Technology, ~Vol.~XX, No.~XX, XXX~2019}
{}

\maketitle

\begin{abstract}
Usage of mobile sink(s) for data gathering in wireless sensor networks(WSNs) improves the performance of WSNs in many respects such as power consumption, lifetime, etc. In some applications, the mobile sink $MS$ travels along a predefined path to collect data from the nearby sensors, which are referred as sub-sinks. Due to the slow speed of the $MS$, the data delivery latency is high. However, optimizing the {\em data gathering schedule}, i.e. optimizing the transmission schedule of the sub-sinks to the $MS$ and the movement speed of the $MS$ can reduce data gathering latency. We formulate two novel optimization problems for data gathering in minimum time. The first problem determines an optimal data gathering schedule of the $MS$ by controlling data transmission schedule and the speed of the $MS$, where the data availabilities of the sub-sinks are given. The second problem generalizes the first, where the data availabilities of the sub-sinks are unknown. Plane sweep algorithms are proposed for finding optimal data gathering schedule and data availabilities of the sub-sinks. The performances of the proposed algorithms are evaluated through simulations. The simulation results reveal that the optimal distribution of data among the sub-sinks together with optimal data gathering schedule improves the data gathering time.

\end{abstract}


\section{Introduction}
\label{sec:intro}


In WSNs, data generated at the sensor nodes are either transmitted through multi-hop transmission to a base station \cite{He:2014,Yao:2015}, or a mobile sink ($MS$) moves through the communication regions of the sensors and collects data from sensors directly/indirectly and brings them to a base station \cite{He:2013, Ma:2013}. In multi-hop transmission, sensors located near the base station are overloaded for relaying data from other sensors to the base station and are therefore prone to deplete their energy faster than other far away sensors.


Recently, mobile sink based data gathering has been gaining popularity significantly in wireless sensor networks (WSNs). In some applications, the MS periodically patrols the sensors, collects their data, returns to the base station and dumps the collected data at the base station. The problem of determining the tour of the $MS$ has been studied rigorously in \cite{He:2013, Ma:2013, Cheng:2016}. Mobile sink based data gathering improves the performance of WSNs in terms of energy consumption and lifetime of the sensors. However, introducing $MS$ as a data carrier in the network increases the data delivery latency due to the slow speed of the $MS$. Reducing the data delivery latency is a critical issue for the $MS$ based data gathering. Ren and Liang et al. \cite{Ren:2013} have shown that the volume of data collection is proportional to the data delivery latency. Time-sensitive applications such as forest fire detection, intrusion detection etc., demand time bound data delivery. Thus, improving data collection with minimum delivery latency is one of the most challenging issues in $MS$ based data gathering. 

Several data gathering algorithms are proposed to improve the data gathering time using mobile sink by shortening the tour length of the $MS$  \cite{He:2013, Ma:2013, Cheng:2016}. The data gathering time depends on the speed of the $MS$ and the length of the tour. There are some studies in \cite{ Sayyed:2015, Huang:2016} on the adaptive speed planning of $MS$ along a predefined path. The $MS$ adjust its speed to maximize network utility and minimize energy consumption. Data gathering problems for rechargeable sensor networks are formulated in \cite{Zhang:2017, Zhang:2016, Guo:2014} by jointly optimizing mobile data gathering and energy provisioning. Gao et. al. \cite{Gao:2010, Gao:2011} propose novel data collection scheme, where a $MS$ is moving along a predefined path with a fixed speed. But, the $MS$ gets limited communication time to collect data from its nearby sensor nodes, referred as sub-sinks. Besides, a metaheuristic (genetic) algorithm is proposed to find data forwarding paths to improve the network throughput as well as to conserve energy. Due to the non-deterministic nature of the algorithm, the solution may vary each time you run the algorithm on the same instance. Therefore, the existing data gathering techniques using $MS$ find optimal tour of $MS$ or find data forwarding paths to $MS$ to improve the network performance, but there is a lack of studies on how to maximize data collection and minimize the data gathering time by controlling the data transmission schedule and speed of the $MS$. The data transmission schedule of the sub-sinks to the $MS$ together with the speed schedule of the $MS$ is called as {\em data gathering schedule} of $MS$. To further improve the total data gathering time, we consider the above two factors and find an optimal distribution of the data generated within sensors among the sub-sinks. In addition, our algorithms are based on the geometric characteristics of the problem and are deterministic. Their correctnesses are also shown.


\begin{figure}
\centering 
\includegraphics[width=8cm]{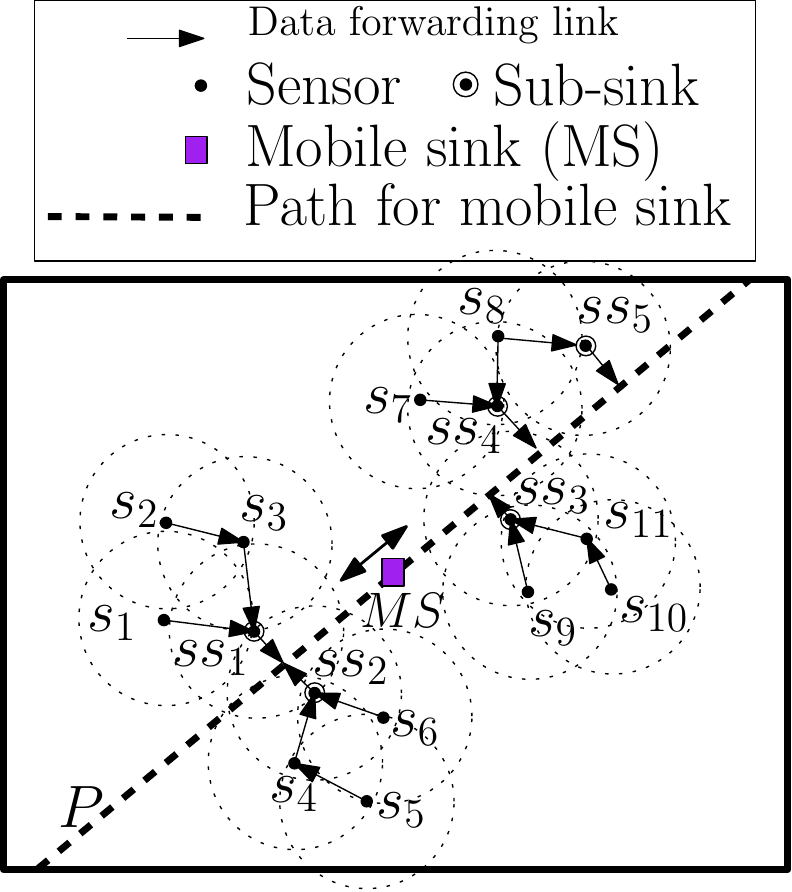}
\caption{Example of Path-constrained mobile sink based sensor network.}
\label {fig:ex1}
\end{figure}

An example of such type of network is illustrated in Figure \ref{fig:ex1}. A mobile sink $MS$ moves along a given path $P$. It collects pre-cached data from a few sensors which are directly reachable from the trajectory path $P$. Those directly reachable sensors are referred as sub-sinks ($ss_1, ss_2, ss_3, ss_4, ss_5$). The $MS$ may collect data from a sub-sink whenever the $MS$ comes under the communication range of the sub-sink. Thus, the sub-sinks send their data to the $MS$ directly. The $MS$ may be within multiple sub-sinks' communication regions, and it receives data from any one of them at a time. Therefore, proper data transmission schedule of the sub-sinks are also required. The remaining sensors, which are not directly reachable to $MS$ ( e.g. $s_1, s_2, \ldots, s_{11}$) send their data to the $MS$ through the sub-sinks using multi-hop communication. The major challenges are to find the optimal data transmission schedule of the sub-sinks to the $MS$ for their data delivery and  speed variation of $MS$ along $P$. Moreover, since a sensor can send data to $MS$ through multiple sub-sinks, finding the optimal data distribution among the sub-sinks is another challenging issue in $MS$ based data gathering. Our major contributions in this article are summarized as follows.

\begin{enumerate}

\item Introduce a time-sensitive data gathering problem using a speed adjustable mobile sink to collect data from sensor networks.

\item  Linear programming formulation of the problem is discussed, where the initial data availabilities of the sub-sinks are given.

\item Plane sweep based data gathering algorithm is proposed to collect data from the sub-sinks by controlling the data transmission schedule of sub-sinks and speed of the $MS$.

\item It is further generalized, where data availabilities of the sub-sinks are optimized by controlling sensors' data distribution among the sub-sinks to improve the data gathering time.



\end{enumerate}


The rest of the article is organized as follows. Section \ref{sec:relatedworks} discusses some related works on data gathering problems. Section \ref{sec:problemstatement} presents system model and problem statement. Background and related terminologies are defined in section \ref{sec: background}. Section \ref{sec:known_data_availability} describes a plane sweep algorithm for data gathering in minimum time, where data availabilities of the sub-sinks are given. Section \ref{sec:unknown_data_availability} presents a plane sweep algorithm to improve the data gathering time by optimizing the data availability values of the sub-sinks so that the total data gathering time can be reduced. Section \ref{sec:experiment} measures the performance of our proposed solutions. Finally, section \ref{sec:conclusion} concludes the article.


\section{Related Works} \label{sec:relatedworks}

Several data gathering algorithms have been proposed in the literature using mobile sink $MS$ where the path of the $MS$ is controllable or fixed. Depending on applications, different objectives are attained such as maximizing network lifetime, minimizing the total energy consumption, reducing total tour length, etc. In this section, we classify the literature based on whether the path of the $MS$ is controllable or fixed. 



Somasundara et. al. \cite{Somasundara:2007} claim that the sensors with higher variation in sensed data demand more frequent data collection than others. They proposed a solution based on optimizing travelling path of the $MS$ that allows the $MS$ to visit sensor with a different frequency to reduce buffer overflow. They prove that the decision version of the problem is NP-complete and two heuristic algorithms are proposed. The authors in \cite{Sugihara:2010} analyze various models of motion planning of mobile sink to solve mobile sink scheduling problem in order to minimize the data delivery latency of the network. He et. al. formulate the data gathering problem using $MS$ as a travelling salesman problem (TSP) with neighborhoods \cite{He:2013}. They schedule $MS$ through the deployed region to improve the tour length of the $MS$ and consider multi-rate wireless communication for data transmission. In \cite{Ma:2013}, a periodic data gathering protocol is proposed for a disconnected sensor network. The $MS$ traverses the entire sensor network, polls sensors and gathers sensed data from sensors. It improves the scalability issue of large-scale sensor networks. Sayyed et al. in \cite{Sayyed:2015} investigate the utility of speed control mobile sink for collecting data in WSN. Single-hop clustering technique is used to increase the data collection rate as well as to decrease the data collection latency.


To overcome the delay due to the slow speed of the $MS$, a subset of sensors are selected as rendezvous nodes. These nodes are used to buffer the data temporarily from the nearby sensors. When the $MS$ visits these rendezvous nodes, then they transfer their data to the mobile sink. In \cite{Alhasanat:2014}, sensors are grouped into single hop clusters, and  the mobile sink visits the centroids of these clusters. If the tour length of the set of centroids is greater than a given upper bound, then some of the clusters are removed until the tour length is less than the upper bound. In \cite{Almiani:2016}, a shortest path tree rooted at the initial position of the mobile sink is built, and then a sensor node having sufficient energy as well as many nearby sensors within its vicinity is chosen as the next rendezvous node. In this way, a set of rendezvous nodes are selected and then a travelling salesman tour is obtained over the selected rendezvous nodes. In \cite{Kaswan:2017}, k-means clustering with a weight function is used for finding the rendezvous points and an efficient tour among the rendezvous points is determined for the $MS$. In addition, an efficient data gathering scheme is also proposed to reduce the total packet drop. In \cite{Trapasiya:2017}, rendezvous nodes are selected using set covering problem. The $MS$ tour is scheduled to pass through those rendezvous points. They introduce novel rendezvous node rotation scheme for fair utilization of all the nodes.  Konstantopoulos et. al. in \cite{Konstantopoulos:2018} use multiple mobile sinks to ensure timely delivery of data to the base station. Mobile sinks visit only a subset of rendezvous points while the remaining sensors forward their data to the rendezvous points through multi-hop communication. The proposed approach increases network lifetime by finding tour passing through energy-rich zones as well as through regions where energy consumption is high.


In some scenarios, the trajectory of the $MS$ is predefined to a fixed path. Efficient data collection algorithms are proposed to improve network performance.  In \cite{Gao:2010, Gao:2011}, data gathering algorithms are proposed for such cases to improve network performance. Data forwarding paths from the sensors to the sub-sinks are determined to maximize the data collection and balance the energy consumption. Huang et. al. in \cite{Huang:2016} consider a scenario where a label of importance is assigned to each sensing region. A path-constrained ground vehicle with adaptive speed is used to collect data from the sensing field. Although the approach tries to improve the data collection throughput, their speed control algorithms are reactive due to the adaptive nature of speed learning characteristics of the $MS$. Besides, these algorithms don't have any specific solution for controlling or optimizing the speed of the $MS$ to improve the data collection rate and minimize delay. In article \cite{Kumar:2017}, a deterministic algorithm is proposed for maximization of data collection using fixed speed mobile sink. However, it may not provide quality data collection due to the mismatch between the data available to the gateways and the data communication time between the gateways and $MS$. Maximizing the data collection throughput in rechargeable sensor networks is addressed in \cite{Mehrabi:2016}. Zhang et. al. \cite{Zhang:2017} maximize data collection while maintaining the fairness of the network in rechargeable sensor networks.  In \cite{Dash:2018, Kumar:2018}, data gathering protocols are proposed from path-constrained mobile sensors. The major drawback of the $MS$ based system is its slow speed, which causes long data gathering delay. Since sensors have limited memory, it causes buffer overflow in the sensors. To avoid buffer overflow, multiple mobile sinks are deployed and they periodically collect data from the mobile sensors and deliver the collected data to the base station.



 
 It can be noted that several data gathering techniques have been proposed which focus on reducing the data gathering time of the mobile sink. The existing literature on path constrained mobile sink mostly consider efficient data forwarding mechanism from the sensors to the mobile sink through the sub-sinks to improve the network performances. But, no existing works consider controlling the data transmission schedule of the sub-sinks to the $MS$ along with the speed of the $MS$ and the sensor's data distribution among sub-sinks to improve the total data collection and the total data gathering time of the mobile sink.



\section{System Model and Problem Formulation} \label{sec:problemstatement}

We consider a wireless sensor network (WSN) which consists of a set of sensors $N = \{s_1, s_2, \ldots, s_{n} \}$. Sensor $s_i$ generates/senses $DG(s_i)$ amount of data from its environment. The communication topology of the network is modelled as an undirected graph $G (N, E)$. The communication regions of the sensors are modelled as disks. There is a mobile sink $MS$ moving on a given path $P$. We assume that the path $P$ is approximated as piecewise straight line segments. The $MS$ can move with a given maximum speed value $V$ to collect data from the sub-sinks. However, the $MS$ can change its speed depending upon the data availabilities of the sub-sinks. The $MS$ can collect data from sensors whose communication disks intersect the path $P$.  Based on the relative position of the sensors with respect to $P$, sensors are divided into two groups, sub-sinks and far-away sensors. Sensors which can directly communicate with $MS$ on $P$ are referred as {\em sub-sinks} and rest of the sensors are referred as {\em far-away} sensors. The far-away sensors send their data to $MS$ through the sub-sinks. Let $\mathbb{SS} = \{ss_1, ss_2, \ldots, ss_m\}$ represent a set of sub-sinks which is a subset of $N$.  


Furthermore, we also assume that the $MS$ and the sensors have sufficient energy and memory to collect and store all the sensed/relayed data temporarily.  The data delivery capacity of a sub-sink is the amount of data that can be delivered by the sub-sink to the $MS$. The data delivery capacity of a sub-sink $ss_i$ depends on the time $t^i$ the $MS$ allocates to $ss_i$ for its data delivery within the communication region of $ss_i$ and the data transmission rate $dtr$. We assume that the $MS$ can receive data from one sub-sink at a time. We also assume that there is no data aggregation in the network. Then, our problems are stated as follows.


{\bf Problem 1:} Let the data availabilities of the sub-sinks be $DA = \{DA(ss_1), DA(ss_2), \ldots, DA(ss_m) \}$. Our objective is to find data transmission schedule of the sub-sinks to the $MS$  and a speed-schedule of the $MS$ through $P$ such that the $MS$ can collect complete data from all the sub-sinks in minimum time.


Our second problem generalizes the previous one, where data gathering time is further improved by optimizing the data availabilities of the sub-sinks.

{\bf Problem 2:} Find an optimal data availabilities of the sub-sinks $DA = \{DA(ss_1), DA(ss_2), \ldots, DA(ss_m)\}$ by distributing the sensors' data among the sub-sinks along with their data transmission schedule to the $MS$ and the speed-schedule of the $MS$ through $P$ such that the $MS$ can collect complete data from all the sub-sinks in minimum time.


\section{Background and Terminologies} \label{sec: background}


The far-away sensors send their data to the $MS$ through the sub-sinks. A sub-sink generates its data and receives data from other sensors and store them temporarily in its local buffer. This buffered data is delivered by the sub-sink to the $MS$ when it passes through the sub-sink's communication region. We refer this buffered data as {\em data availability} of the sub-sink.

Since the maximum speed $V$ of the $MS$ is given, a naive approach for the $MS$ is to move at this maximum speed $V$ on $P$ and visit all the sub-sinks and collect their data. But it may not collect complete data from all the sub-sinks. However, if the speed of the $MS$ can be varied according to the data availabilities of the sub-sinks, then it may improve the amount of data collection. 

For instance, the $MS$ should move at slow speed within the communication range of a sub-sink, which has more data, whereas it should move at a faster speed within the communication range of a sub-sink which has less or no data. Furthermore, the $MS$ should move with its maximum speed of $V$, when it is not under the communication range of any sub-sink or the sub-sinks do not have data to deliver. Determining the speed of the $MS$ at different position on $P$ is referred as  {\em speed-schedule} of $MS$. Note that the speed of the $MS$ can vary between $0$ to $V$. It may happen that the $MS$ is within multiple sub-sinks communication regions then one of the sub-sinks can transmit data to the $MS$. Therefore, proper time sharing among the sub-sinks is also required. We refer it as {\em data transmission schedule} of the sub-sinks. The {\em data transmission schedule} of the sub-sinks to the $MS$ together with the {\em speed schedule} of the $MS$ is called as {\em data gathering schedule} of $MS$. Our objectives are to find optimal {\em data gathering schedule} of the $MS$ through the communication regions of the sub-sinks along the path $P$ to collect complete data from the sub-sinks in minimum time. 

A possible speed-schedule for $MS$ is shown in Figure \ref{fig:ex2}. Figure \ref{fig:ex2}(a) shows the path $P$ of the $MS$ with dashed line and circles denote the communication disks of the sub-sinks. Let the $MS$ start from $S$ and end at $E$ while travelling through the path $P$. Figure \ref{fig:ex2}(b) shows the speed-schedule of the mobile-sink at different position on $P$. It shows that the speed of $MS$ is slow within the communication disks of the sub-sinks whereas it runs with its maximum speed $V$ outside the communication disks.





\begin{figure}[!t]
\centering 
\includegraphics[width=8cm]{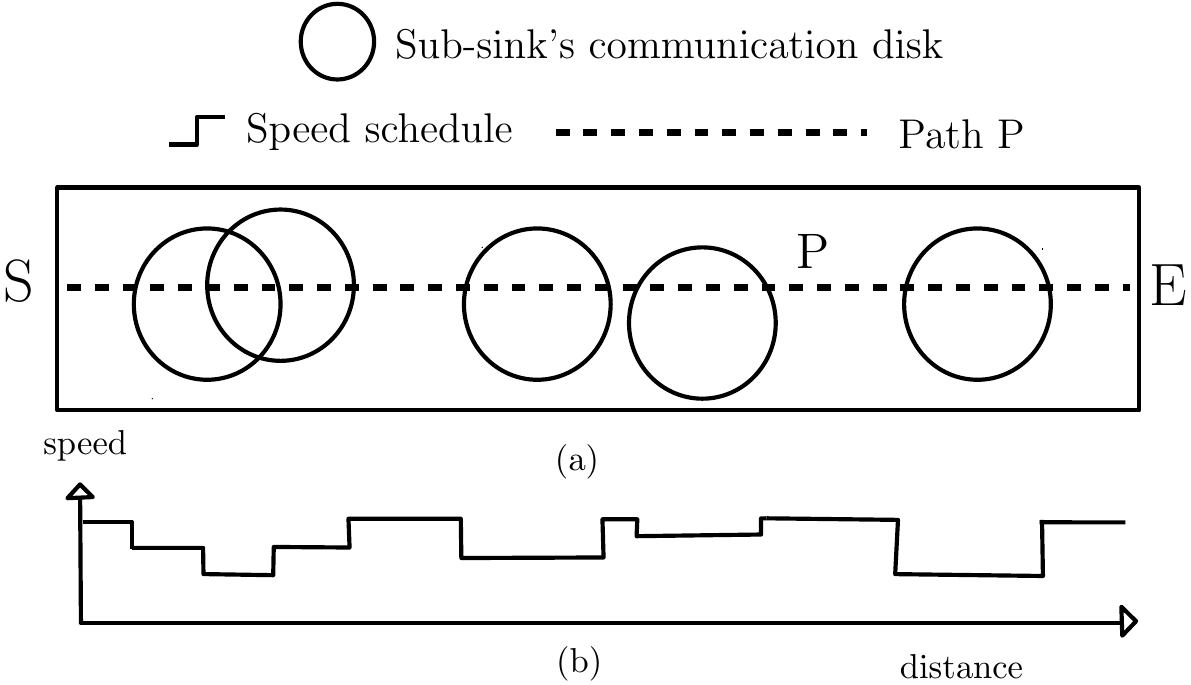}
\caption{Speed-schedule of mobile sink in sensor network}
\label {fig:ex2}
\end{figure}

We introduce some terminologies to describe our algorithm, which are as follows.

\begin{definition}{\bf Start-point ($p_i^s$) :}
It is a first point on $P$ from which a sub-sink $ss_i$ can communicate or start delivering data to the $MS$. 
\end{definition}

\begin{definition}{\bf End-point ($p_i^e$) :}
It is a last point on $P$ after which a sub-sink $ss_i$ cannot communicate or ends delivering  data to the $MS$. 
\end{definition}


\begin{definition} {\bf Data availability ($DA(ss_i)$) : }
It is the amount of data available at a sub-sink $ss_i$. 
\end{definition}

\begin{definition} {\bf Data delivery time ($DT(ss_i)$) : }
It is the minimum time requirement to transmit the data available at a sub-sink $ss_i$ to the $MS$. 
\end{definition}

Data delivery time is determined using  $DT(ss_i)$ = $\frac {DA(ss_i)}{dtr}$ formula, where $dtr$ denotes the data transmission rate between $ss_i$ and $MS$. 






\section{Data Gathering in Minimum Time (data availability of sub-sinks are known apriori) }
\label{sec:known_data_availability}

In this section, we first discuss linear programming problem (LPP) formulation of the proposed problem,  thereafter we discuss a plane sweep based algorithm. The mobile sink $MS$ travels through the path $P$ and collects complete data from all the sub-sinks. The data availability values of the sub-sinks are given. The $MS$ receives data from one sub-sink at a time. The objective is to collect the complete data from all the sub-sinks in minimum time.  We control the {\em data gathering schedule} of the $MS$. In other words, the time allocation of the $MS$ to the sub-sinks and the time spent by the $MS$ within their communication regions are determined based on their data availability values to minimize the data gathering process.  

\subsection{LPP Formulation}

The ordering of the start-points and end-points of the sub-sinks partition the path $P$ into disjoint segments/intervals. An example of partitioning the path $P$ into segments is shown in Figure \ref{fig:ex3}. In Figure \ref{fig:ex3}(a), a set of sub-sinks $\mathbb{SS}=\{ ss_1, ss_2 \ldots ss_5 \}$ and their start-points and end-points are shown on the path $P$. The ordering of the start-points and end-points of the sub-sinks partition the path $P$ into disjoint segments, which are shown in Figure \ref{fig:ex3}(b). Zero or more sub-sinks are reachable to the $MS$ from a particular segment. The idea of the solution is that the data gathering time within each segment is shared properly among the sub-sinks such that the $MS$ can collect complete data from all the sub-sinks through the segments and total data gathering time from the starting position $S$ to the ending position $E$ is minimum. Also, the $MS$ maintains the maximum speed limit constraint. 

\begin{figure}
\centering 
\includegraphics[width=8cm]{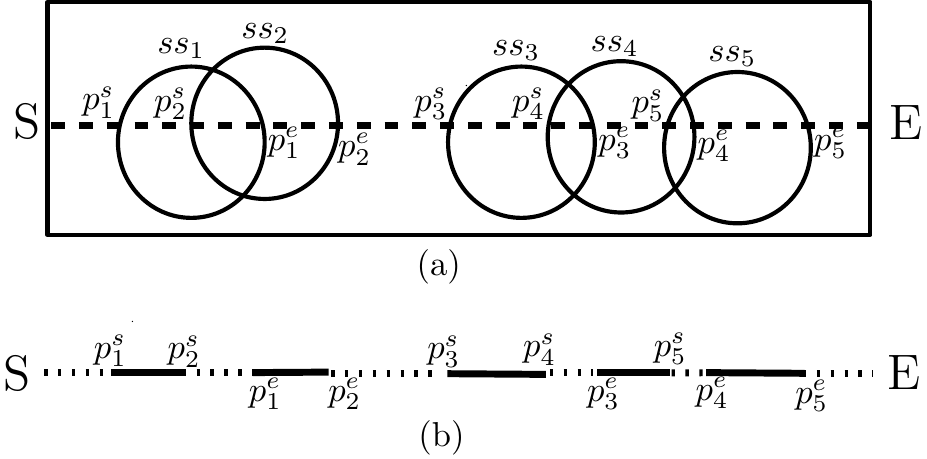}
\caption{Start-point and end-point of sub-sinks on path $P$}
\label {fig:ex3}
\end{figure}

The start-points and end-points of the sub-sinks partition the path $P$ into disjoint segments/intervals. The $m$ sub-sinks have $2m$ end-points. This will partition the path $P$ into at most $2m+1$ disjoint segments $\{ I_1, I_2, \ldots I_{2m+1} \}$. For each segment, we use a set of variables for the set of sub-sinks reachable from the segment. Within a particular segment $I_j$, the set of sub-sinks reachable to the $MS$ remains unchanged. Let $SS(I_j)$ denote the sub-sinks in $\mathbb{SS}$ reachable to the $MS$ within the segment $I_j$. Let $t_j^i$ denote the time allocated to sub-sink $ss_i \in SS(I_j)$ for transferring its data to the $MS$ on the segment $I_j$. If a segment $I_j$ is not reachable to any sub-sink, then we assume that it is reachable from a virtual sub-sink $ss_0$ which has no data, i.e. $DA(ss_0)=0$, and the time the $MS$ spends to cross the segment $I_j$ is denoted by $T(I_j)=t_j^0 \ge \frac{|I_j|}{V}$. Similarly, if from segment $I_j$ two sub-sinks $ss_i$ and $ss_k$ are reachable, i.e. $SS(I_j) = \{ss_i, ss_k \}$, then there are two variables $t_j^i$ and $t_j^k$ corresponding to two sub-sinks for the segment $I_j$. Each variable value denotes the amount of time allocated to the corresponding sub-sink for data delivery when the $MS$ travels through the segment $I_j$. 

There are two types of constraints (i) time spent on each segment $I_j$ by the $MS$ is at least the travelling time $\frac{|I_j|}{V}$, and (ii) total time $t^i = \sum_{j=1}^{2m+1}{ t_j^i } : ss_i \in SS(I_j)$, allocated by the $MS$ to a sub-sink $ss_i$ for its data delivery, must be greater than or equal to the sub-sink's data delivery time $DT(ss_i)$. The LPP formulation of the said problem is shown in Equation \ref{equ:prob1}.


\begin{equation} \label{equ:prob1}
\begin{aligned}
\text{Minimize : }   & { \sum \limits_{j=1}^{2m+1}{\sum \limits_{ i : ss_i \in SS(I_j) }{ t_j^i }} }     \\
\text{Subject to : } & { \sum \limits_{ i  :  ss_i \in SS(I_j) }{ t_j^i } \ge \frac{|I_j|}{V} , }  \quad  j=   1 \dots (2m+1)  \\
                            & { \sum \limits_{j  :  ss_i \in SS(I_j) }{ t_j^i } \ge DT(ss_i) , }  \quad  i=   1 \dots m \\
                            & { t_j^i  \ge 0, \quad   i=0 \dots m,  \quad j=  1 \dots (2m+1) } 
\end{aligned}
\end{equation}


After solving the LPP in Equation \ref{equ:prob1},  $t_j^i , \quad i= 0 \dots m,  j=  1 \dots (2m+1) $ are known, which denote the data transmission schedule of the sub-sinks. The lengths of the segments $I_i,  i=1 \dots (2m+1)$ are already derived from the start-points and end-points. Hence, the speed of the $MS$ at different segments can be determined easily. The following subsection discusses a plane sweep based algorithm for the problem.

\subsection{Plane Sweep Algorithm}

The mobile sink $MS$ moves through the path $P$. When the $MS$ is within the multiple sub-sinks' communication range, then the $MS$ receives data from only one of them by prioritizing them according to their end-points positions on $P$. The sub-sink whose end-point appears first on $P$ has higher priority than that sub-sink whose end-point appears later. Let $PR(ss_i)$ denote the priority of a sub-sink $ss_i$.

In the plane sweep algorithm, it is simulated by moving a sweep line through the path $P$. We consider a horizontal data gathering path $P$ for the $MS$ and a virtual vertical line perpendicular to $P$, called {\em sweep line} moves (sweeps) through the path $P$ from $S$ to $E$. While sweeping the sweep line intersects the sub-sinks' communication disks. We have defined two types of events : start-point event and end-point event for every sub-sink. Start-points and end-points of the sub-sinks are stored in an event queue $Q$ according to their appearance on $P$ from left to right. At a particular position of the sweep line on $P$, we maintain a list of sub-sinks in a {\em status line} data structure $L$.  The sub-sinks whose communication disks intersect the sweep line on $P$ are in $L$.  At a particular position on $P$, if multiple sub-sinks' communication disks intersect the sweep line on $P$ and they have data, then a sub-sink $ss_i$ in $L$ with maximum priority gets the preference to deliver data to $MS$.

Initially, all the start-points and end-points of the sub-sinks are added to the event queue $Q$. We are calling three methods to perform different operations on the event queue $Q$. $InsertInQ()$ method is used for inserting an event, $RemoveFromQ()$ method removes the leftmost event on $P$, and $PeekFromQ()$ method retrieves the leftmost event but does not remove it from the queue. Similarly, three methods $InsertInL()$, $RemoveFromL()$, and $PeekFromL()$ are used to perform three different operations on the status line data structure $L$. Events are processed one by one from the event queue $Q$, as the {\em sweep line} moves through the path $P$. The top event is removed from $Q$ and is referred as current event $CE$. A sub-sink $ss_i$ is inserted into $L$, whenever the sweep line processes its start-point $p_i^s$. If the sweep line is processing an end-point $p_i^e$ of sub-sink $ss_i$ and the complete data of $ss_i$ is not yet delivered, then the $MS$ waits at the end-point $p_i^e$ and receives the remaining data from $ss_i$. Subsequently, the sub-sink $ss_i$ is removed from $L$. Thereafter, the next event point $NE$ is picked from the event queue $Q$. Travel time $TT$ of the $MS$ between current event $CE$ and the next event $NE$ is determined assuming that the $MS$ moves with its maximum speed $V$ in between the two events. Thereafter, maximum priority sub-sink $ss_j$ is picked from $L$. If the data transmission time $DT(ss_j)$ of the sub-sink $ss_j$ is $ \le TT$, then the sub-sink $ss_j$ completes data delivery to the $MS$ between the two events. The sub-sink $ss_j$ is removed from $L$. Subsequently, the next highest priority sub-sink in $L$ is picked for data delivery. This process continues until the sweep line reaches another event point or the data delivery process is completed. If there is no sub-sink in $L$ having data to deliver then the $MS$ moves with its maximum speed $V$. The detailed algorithm is presented in Algorithm \ref{Algo:1}. 



In Figure \ref{fig:ex3}(a), the $MS$ starts its journey from $S$ with speed $V$. As it reaches $p_1^s$, then the sub-sink $ss_1$ is inserted into $L$. Thereafter, $MS$ starts receiving data form $ss_1$ until it reaches $p_2^s$. If the data delivery of $ss_1$ is not over, then there are two sub-sinks $ss_1$, $ss_2$ reachable to $MS$ within segment $[p_2^s p_1^e]$. As the end-point $p_1^e$ appears before $p_2^e$, therefore, according to our algorithm, sub-sink $ss_1$ gets the privilege to deliver its remaining data to $MS$ within the segment $[p_2^s p_1^e]$. If the data delivery of $ss_1$ is still not over within $\frac{|p_1^s p_1^e|}{V}$ time, then the $MS$ waits at point $p_1^e$ for the remaining data delivery time for the duration of $DT(ss_1) - \frac{|p_1^s p_1^e|}{V}$ time. Otherwise, the $MS$ starts receiving data from $ss_2$ after crossing the start-point $p_2^s$ and allocating $DT(ss_1)$ time to $ss_1$. In this way, the $MS$ either moves with its maximum speed or waits at the end-points of the sub-sinks until it reaches the end of the path $E$.



\begin{algorithm}
\caption{Plane sweep algorithm for data gathering using $MS$}
\label{Algo:1}

\KwData{ Location($ss_i$) and $DA(ss_i)$ $\forall{ss_i \in \mathbb{SS} }$, $P$, $V$, $dtr$}
\KwResult{ Data Transmission Schedule of the sub-sinks, and Speed Schedule of $MS$}

$\forall{ss_i \in \mathbb{SS} }$ : Compute start-point ($p_i^s$) and end-point ($p_i^e$)  with respect to $P$ \;


$Q= \emptyset$, $m= |\mathbb{SS}|$ \;

\tcc{Initialize event queue $Q$ with start-points and end-points  }

\For {$i =1$ to $m$} {

InsertInQ( $p_i^s$ ); InsertInQ($p_i^e$)\;

 $DT(ss_i)= \frac{DA(ss_i)}{dtr}$;  \tcc{Data delivery time of $ss_i$} 
}

$L=\emptyset$;  \tcc{Initialize status line $L$ }

The $MS$ moves with its maximum speed $V$ from $S$ to the next end-point event, or until it reaches end of the path $E$ \;

\While{ $Q \ne \emptyset $ } {

	$CE = RemoveFromQ()$  \;  
	
	\uIf{ $CE$ = $p_i^s$ } {
	
		InsertInL($ss_i$ )  \;
	}
	
	\uElseIf{ $CE$ = $p_i^e $  $\land$  $DT(ss_i) >0$ } {
		
		$MS$ stops and receives remaining data of $ss_i$ \; 
		$DT(ss_i)=0$  \;

		RemoveFromL($ss_i$) ; 
		
	}
	
	\tcc{ If $L \ne \emptyset$ then select a sub-sink $ss_j$ with maximum priority from $L$ and $MS$ starts receiving data from $ss_j$ }
	
	$NE=PeekFromQ()$ ;  \tcc{ next event}
	
	$TT = \frac{dist(CE,NE) }{V}$ ;   \tcc{ travel time between $CE$ and $NE$ }
	
	$ss_j=  PeekFromL()$;  
	
	\While{ $L \ne \emptyset$  $\land$  $DT(ss_j) \le TT$ }{
	
		  $MS$ moves with speed $V$ and receives data from  $ss_j$ for $DT(ss_j)$ time  \;
	
	           $TT = TT - DT(ss_j)$ \;
		
	           $DT(ss_j) =  0$ \;
	
	           $RemoveFromQ(p_j^e)$ \;
	
	           $RemoveFromL(ss_j)$ \;
	
	           $ss_j=  PeekFromL()$ \;
	
	         }

	          \uIf { $L \ne \emptyset$  $\land$  $DT(ss_j) > TT$  }{
		 
			$DT(ss_j) = DT(ss_j) - TT$ \;

			$MS$ moves with speed $V$ and continue receiving data from  $ss_j$ for $TT$ time \;
		      }
		 \uElse {

		         $MS$ moves with speed $V$ without receiving any data to next event for $TT$ time \;
		        
			 }

}  

\end{algorithm}

\begin{corollary} 
Data gathering sub-paths of the mobile sink $MS$ on $P$ from a sub-sink $ss_i$ is confined within $[p_i^s, p_i^e]$ for $i \in \{1 \ldots m \}$. 
\label{cor:cor1}
\end{corollary}


\begin{theorem}
If the mobile sink $MS$ follows the Algorithm \ref{Algo:1} for data gathering, then it receives complete data from all the sub-sinks. 
\label{th:th1}
\end{theorem}

\begin{proof}
Algorithm \ref{Algo:1} selects the highest priority sub-sink in $L$ for data delivery to the $MS$. A sub-sink $ss_i \in \mathbb{SS}$ is removed from $L$ only when the $MS$ finishes receiving its data by allocating $DT(ss_i)$ time to $ss_i$. The time allocation may be continuous or discontinuous. A sub-sinks $ss_i$ is inserted to $L$ whenever the $MS$ crosses $p_i^s$. Since all the star-points and end-points of the sub-sinks are in $Q$ and are processed. Therefore, all the sub-sinks get a chance to be in $L$. Once the algorithm ends then the event queue $Q$ and the list $L$ become empty. Therefore, all the sub-sinks must have delivered their complete data to the $MS$.  
\end{proof}


\begin{theorem}
The mobile sink $MS$ completes the data gathering process in minimum time by following Algorithm \ref{Algo:1}. \label{th:th2}
\end{theorem}

\begin{proof}
Assume for the sake of contradiction that the $MS$ does not complete the data gathering process in minimum time. According to our algorithm, the $MS$ moves with its maximum speed $V$ throughout the path except at some end-points. So,  there is an extra delay at some end-points. Extra delay for receiving data from a sub-sink $ss_i$ is possible only when the $MS$ waits at $p_i^e$, but for some sub-path of $[p_i^s, p_i^e]$, the $MS$ moves without receiving data from any sub-sink or receives data from a sub-sink $ss_j$, whose end-point  $p_j^e$ appears after $p_i^e$. This is because the sub-sinks, whose end-points appear after $p_i^e$ can deliver data beyond $[p_i^s, p_i^e]$, and may overall reduce the waiting time at $p_i^e$.



According to Algorithm \ref{Algo:1}, once the $MS$ enters $[p_i^s, p_i^e]$, it either receives data from $ss_i$, or any other sub-sink $ss_j$ such that $PR(ss_j) \ge PR(ss_i)$ in $L$. This implies $p_j^e$ appears before $p_i^e$.  Therefore, there is no sub-path within $[p_i^s, p_i^e]$ where the $MS$ moves/waits without receiving data from any sub-sink $ss_j \in L$, where $PR(ss_j) \ge PR(ss_i)$ and waits at $p_i^e$. Hence, the $MS$ does not make extra delay at any end-point and completes the data gathering process in minimum time.
\end{proof}

\begin{theorem} 
Time complexity of the plane sweep algorithm \ref{Algo:1} is $O(m \log{m})$.  \\
\label{th:timecomplexity1}
\end{theorem}

\begin{proof}
Throughout the algorithm, an event point (start-point/end-point) of a sub-sink is inserted once and removed once in the event queue $Q$, and in total $2m$ event points are processed.  The events are processed from event queue $Q$ using a heap data structure. Inserting and then removing the event points require $O(m \log{m})$ time. During the processing of an event, some basic operations on the status line data structure $L$ are performed. In the worst case, $m$ sub-sinks are simultaneously in $L$. Therefore, the time needed to perform an insert or delete operation on the status line is $O(\log{m})$ and the peek operation takes $O(1)$ time.   

The plane sweep algorithm processes $2m$ event points for $m$ sub-sinks. In total $m$ insert and $m$ delete operations, and at most $2m$ peek operations are performed on the status line data structure $L$, and each such operation takes at most $O(\log{m})$ time.  Hence, it follows that the total time processing all the events is $O(m\log{m})$.
\end{proof}


\section{Improving the Data Gathering Time By Optimizing The Data Availabilities of the Sub-sinks}
\label{sec:unknown_data_availability}

The solution in the previous section finds a {\em data gathering schedule} of the $MS$, where the data availability values of the sub-sinks are given. This section generalizes the problem, where data availabilities of the sub-sinks are determined to improve the data gathering time. Proper distribution of the sensors' data among the sub-sinks is carried out to improve the data gathering time. Determining an optimal data distribution among the sub-sinks is another challenging issue in WSN. The data availability values of the sub-sinks are determined using a plane sweep algorithm for the given sensor network. After determining the optimal data availability values of the sub-sinks, we consider the values as data delivery capacity of the sub-sinks and the sensors' data are pushed to those sub-sinks using network flow algorithm. Thereafter, Algorithm \ref{Algo:1} is used to complete the data gathering process in minimum time. In summary, this section discusses the solution for the Problem 2, where our objective is to distribute the sensors' data among the sub-sinks properly so that the $MS$ can collect complete data from all the sub-sinks in minimum time.

\subsection{Determining Data Availabilities of the Sub-Sinks Using Plane Sweep Algorithm} 
\label{subsec:data_availability}

In this subsection, we determine the data availability values of the sub-sinks for a given network topology. We assume that the data generated on the sensors are known, which are denoted as $DG(s_i): i=1:n$. Using a plane sweep algorithm, we determine the data availability values of the sub-sinks $DA(ss_i): i=1:m$. Initially, the sensor network is partitioned into connected components $C=\{ c_1, c_2, \ldots  c_k \}$ based on its communication topology $G$. The idea of this algorithm is that the data generated in a component is distributed among its corresponding sub-sinks so that the $MS$ can collect complete data from the component through its sub-sinks in minimum time. The $MS$ moves with its maximum speed $V$ through $P$, except at a few end-points. 


For individual connected component, the total data generated by the sensors in the corresponding component is determined. Let $\{ DG(c_1), DG(c_2), \ldots DG(c_k) \}$ denote the data generated in the components. The data availabilities of the sub-sinks are initialized to zero : $DA(ss_1)=0, DA(ss_2)=0, \ldots DA(ss_m)=0$. The start-points and end-points of the sub-sinks are determined. The sub-sinks are labelled with their corresponding component identity. Let $C(ss_i)$ denote the component identity of a sub-sink $ss_i$. The {\em last sub-sink} of a component $c_i$ denoted by $LSS(c_i)$ is a sub-sink, whose end-point appears last on $P$ among all the sub-sinks in $c_i$. For each component $c_i$, identify its last sub-sink $LSS(c_i)$. The start-points and end-points of the sub-sinks are stored in an event queue $Q$ according to their order on the path $P$. The status line data structure $L$ is initialized to $\emptyset$.  

A virtual perpendicular sweep line moves through the path $P$ and process the events one after another from the event queue. At a particular position on the path $P$ of the sweep line, it keeps track of all the sub-sinks in a status line $L$, whose communication disks intersect the sweep line on the path $P$. The priority of a sub-sink $ss_i$ in $L$ is based on the two parameters : its corresponding component's last sub-sink's end-point position, i.e. end-point of $LSS(ss_i)$, and its start-point $p_i^s$ on $P$. If two sub-sinks $ss_i$ and $ss_j$ belong to same component, i.e. $C(ss_i) = C(ss_j)$, then the sub-sink whose start-point appears first on $P$, has higher priority than the other sub-sink. If the two sub-sinks belong to different components, i.e. $C(ss_i) \ne C(ss_j)$ and the end-point of $LSS( C(ss_i) )$ appears before the end-point of $LSS(C(ss_j) )$ on $P$, then $PR(ss_i) > PR(ss_j)$. 

The top event is removed from the queue $Q$ and is referred as the current event $CE$. If $CE$ is a start-point of $ss_i$, then $ss_i$ is inserted into $L$. If $CE$ is an end-point of sub-sink $ss_i$ and it is the last sub-sink of its corresponding component $c_j$ and the component has data  ($DA(c_j) > 0$) then the data availability of $ss_i$ is increased by $DA(c_j)$.  Subsequently, the sub-sink $ss_i$ is removed from $L$. Thereafter, the next event point $NE$ is picked from the event queue $Q$. Travel time $TT$ of the $MS$ between current event $CE$ and the next event $NE$ is determined assuming that the $MS$ moves with its maximum speed $V$ in between the two events. 

Next, the maximum priority sub-sink $ss_j$ is picked from $L$. Let $C(ss_j)$ denote the component of sub-sink $ss_j$. If the remaining data availability of the component $C(ss_j)$, which is $DA(C(ss_j)) \le TT*dtr$ (data transmission capacity of $ss_j$ between the two event points), then data availability of $ss_j$ is increased by $DA(C(ss_j) )$. The sub-sink $ss_j$ is removed from $L$. The remaining travel time between the two events $CE$ and $NE$ of the $MS$ is updated accordingly. This process continues until the remaining travel time by the $MS$ is exhausted and subsequently process the next event. In other words,  if data transfer from a component $c_i$ is over before the sweep line reaches the end-point of its corresponding last sub-sink $LSS(c_i)$, then all the sub-sinks in $c_i$ are removed from $L$. This process continues until the sweep line reaches the next event point or the data delivery process is completed. The detailed algorithm for finding data availabilities of the sub-sinks is shown in Algorithm \ref{Algo:2}.

\begin{algorithm}
\caption{Plane sweep algorithm for computing data availabilities of the sub-sinks }
\label{Algo:2}

\KwData{Communication topology G, Data generated by the sensors $\{DG(s_1), DG(s_2), \ldots DG(s_n) \}$, $\mathbb{SS}$, $P$, $V$, $dtr$ }
\KwResult{Data availabilities of the sub-sinks  $DA=\{ DA(ss_1), DA(ss_2), \ldots DA(ss_m) \}$ }

Partition the sensor network into components $C=$ $\{c_1, c_2, \ldots c_k \}$ based on its communication topology \; 

$\forall {c_i \in C}$  : Compute total data generated $DG(c_i)$ by adding all sensors data in the component \; 

$\forall {c_i \in C}$ : $DA(c_i)= DG(c_i) $  \; 

$\forall {ss_i \in \mathbb{SS}}$ : $DA(ss_i)= 0 $ \; 

$\forall {ss_i \in \mathbb{SS}}$  : Find start-point ($p_i^s$), end-point ($p_i^e$) and component-id $C(ss_i)$ \;

 $\forall c_i \in C$: Find last sub-sink $LSS(c_i)$ \; 


\tcc{Initialize event queue $Q$ with start-point and end-point of the sub-sinks }

$Q= \emptyset$, $m= |\mathbb{SS}|$ \;

\For {$i =1$ to $m$} {
        
	InsertInQ($p_i^s$ );   InsertInQ($p_i^e$)\;
}

$L=\emptyset$ ; 	\tcc{Initialize status line $L$ }


\While{ $Q \ne \emptyset $ } {

	$CE=RemoveFromQ()$ \;    
	
	\uIf{ $CE$ = $p_i^s$ } {
	
		InsertInL($ss_i$ ) \;  
	    }
	\uElseIf{ $CE$= $p_i^e$  } {
		
             \uIf{ $ss_i \in c_j  \land  ss_i = LSS(c_j) \land DA(c_j) >0 $  }
		{
			$DA(ss_i) = DA(ss_i) + DA(c_j) $\;
			$DA(c_j) =0$ \;
		}
	
		RemoveFromL($ss_i$) \; 
	    }


	$NE=PeekFromQ()$ \;    
	
	Let dist(CE,NE) = Distance between events $CE$ and $NE$ \;

	$TT= \frac{dist(CE, NE)}{V}$ ;  \tcc{Travel time between $CE$ and $NE$} 
		
	$ss_j= PeekFromL()$ ;  \tcc{ Peek maximum priority sub-sink in $L$ }  
	
	\While{$L \ne \emptyset$   $\land$  $DA( C(ss_j) )  \le  TT * dtr$ }{
			
		 $DA(ss_j)= DA(ss_j) + DA( C(ss_j) )$  \;
		
		 $DTT = \frac{DA( C(ss_j )  ) }{ dtr }$ ;  \tcc{Data transfer time} 
	       
		$TT= TT - DTT$ \;
		$DA( C(ss_j) )=0$  \;
		
		$RemoveFromQ(p_j^e) $ \;
	      $RemoveFromL(ss_j) $ \;
	
	      $ss_j=  PeekFromL() $ \;
	
	   } 
	
           \uIf{ $L \ne \emptyset$  $\land   DA(C(ss_j)) > TT* dtr$ }{
          	
          		$DA(ss_j) = DA(ss_j) + TT* dtr $  \;
			$DA(C(ss_j)) =  DA(C(ss_j)) - TT* dtr$  \;
	    }
		
}  

\end{algorithm}

\begin{theorem} 
Time complexity of the plane sweep algorithm \ref{Algo:2} is $O(n+ e + m \log{m})$, where $n$ and $e$ denote the number of sensors and number of links in the communication graph $G$.  \\
\label{th:timecomplexity2}
\end{theorem}

\begin{proof}
Depth first search is used for partitioning the network into components which can be performed in $O(n+e)$ time. Computing total data generated for each component can be performed in $O(n)$ time. Finding the start-points and end-points of the sub-sinks can be done in $O(m)$ time. Identifying the last sub-sink for each component can be done in $O(n)$ time. The time complexity analysis for the rest of the algorithm is similar to Algorithm \ref{Algo:1}. The plane sweep algorithm processes $2m$ event points. The events are inserted and then removed from the event queue, which takes overall $O(m \log{m})$ time. During the processing of an event, some basic operations on the status line data structure are performed. There are at most $m$ sub-sinks intersecting the sweep line on $P$ at any time and therefore, the time needed to perform an insert or delete operation on status line is $O(\log{m})$ and peek operation can be performed in $O(1)$ time. Through the algorithm, a sub-sink is inserted once and removed once from the status line data structure. Therefore, the total time spent on accessing the sweep line status data structure is $O(m\log{m})$.  Hence, it follows that the total time spent processing all the events is $O(m \log{m})$. Therefore, the total time complexity of the algorithm is  $O( n + e + m\log{m} )$ 
\end{proof} 

\subsection{Distributing Data Among the Sub-Sinks Using Network Flow Algorithm } 
\label{subsec:data_distribution}

Once the data availability values of the sub-sinks $DA(ss_1), DA(ss_2)$ $\ldots DA(ss_m)$ are determined using Algorithm \ref{Algo:2}, this phase distributes the sensors' data among the sub-sinks. Data are pushed from the sensors to the sub-sinks based on their calculated data availability values. Sensors use the communication topology network to send data to the sub-sinks. Network flow algorithm is used for finding data flow from the sensors to the sub-sinks. Construction of network flow graph and determining the distribution of data from the sensors to the sub-sinks for a given communication topology is described with an example for the sensor network in Figure \ref{fig:ex1}.

\begin{figure}[!t]
\centering 
\includegraphics[width=8cm]{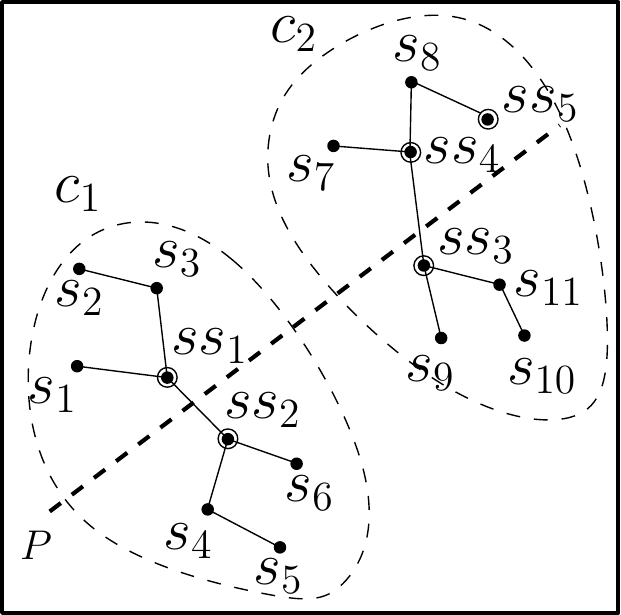}
\caption{Connected components corresponding to the sensors network}
\label {fig:ConnectedComponents}
\end{figure}

The connected components corresponding to the sensor network of Figure \ref{fig:ex1} are identified and labelled with $c_1$, and $c_2$ in Figure \ref{fig:ConnectedComponents}. To determine the data distribution from the sensors to the sub-sinks, a network flow graph is constructed using the communication topology of the sensor network. Thereafter, the sensors' data are distributed among the sub-sinks based on the data availability values of the sub-sinks, the amount of data generated within the sensors, and the communication topology. The network flow graph corresponding to the communication topology in Figure \ref{fig:ConnectedComponents} is shown in Figure \ref{fig:FlowGraph}. A virtual source vertex $VS$  and a virtual sink vertex $VK$ are added to the network topology. To maintain the cleanness of the figure, we have drawn four duplicate virtual source vertices, but actually they are a single vertex $VS$.  The virtual source vertex is incident to the sensor nodes including the sub-sinks using virtual links. The capacities of these virtual links are set based on their data generation capacities. Therefore, the capacity of a link between $VS$ and $s_i$ is $DG(s_i)$. Similarly, the link capacity between $VS$ and $ss_i$ is $DG(ss_i)$ because, as these are data generation limits of the sensor $s_i$ /sub-sink $ss_i$. The sub-sinks are incident to the virtual sink $VK$ through virtual links. The capacity of a virtual link between a sub-sink $ss_i$ and the virtual sink $VK$ is set to $DA(ss_i)$, which is its data availability value determined in the previous phase. Other links represent the communication links among the sensors/sub-sinks, and their capacities are set to infinity because we assume that a sensor can forward the data generated within itself or received from its neighbors. 

Thereafter, the network flow algorithm is used for finding the maximum data flow from the $VS$ to $VK$. The flow value of the links denotes the data flow between the corresponding sensors/sub-sinks. Finally, data is delivered from a sub-sink to virtual sink $VK$. The flow value between a sub-sink and the virtual sink denotes the actual data delivery by the sub-sink to the $MS$. 

\begin{figure}[!t]
\centering 
\includegraphics[width=8cm]{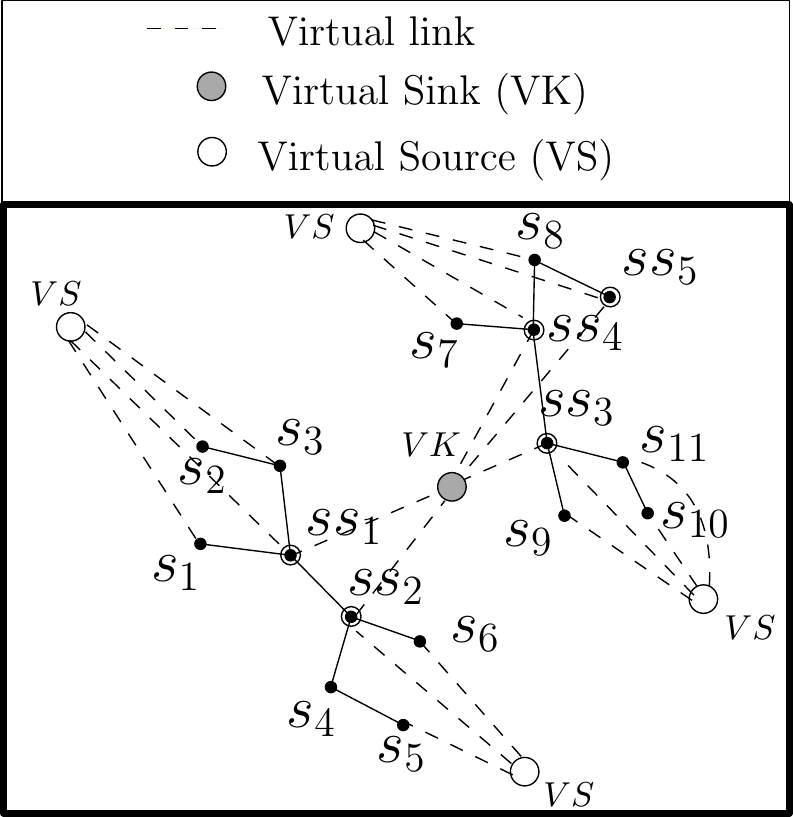}
\caption{Network flow graph corresponding to the sensors network; Link capacity : $(VS, s_i)= DG(s_i)$; $(VS, ss_i)= DG(ss_i)$;  $(ss_i, VK)= DA(ss_i)$;  other links capacities are $\infty$ }
\label {fig:FlowGraph}
\end{figure}

\subsection{Gathering Data Using Algorithm \ref{Algo:1} }
\label{subsec:speed_schedule}

In this subsection, we find an optimal data gathering schedule of the mobile sink ($MS$) to collect complete data from all the sub-sinks in minimum time. Once the data availabilities of the sub-sinks  $DA(ss_1), DA(ss_2) \ldots DA(ss_m)$ are determined, and data are pushed from the sensors to the sub-sinks, we use the Algorithm \ref{Algo:1} of Section \ref{sec:known_data_availability} to find the  {\em data gathering schedule} of the $MS$. 




\begin{theorem}
If the data availabilities of the sub-sinks are determined using Algorithm \ref{Algo:2} and $MS$ follows the speed-schedule using Algorithm \ref{Algo:1}, then the mobile sink $MS$ completes the data gathering process in minimum time. 
\label{th:th4}
\end{theorem}

\begin{proof}
The data availabilities of the sub-sinks are determined using Algorithm \ref{Algo:2} such that the $MS$ is able to receive complete data from a component while moving with its maximum speed $V$ and if required waits only at the last sub-sink's end-point. The data generated in the sensors are distributed among its sub-sinks based on the data availability values determined using Algorithm \ref{Algo:2}. According to Algorithm \ref{Algo:1}, while the $MS$ is moving, it receives data from the highest priority sub-sink having data to deliver. Let the first sub-sink of a component $c_i$ be a sub-sink $ss_j \in c_i$, whose start-point $p_j^s$ appears first on $P$. Let $c_i^s =p_j^s$ denote the start-point of the first sub-sink of $c_i$. Similarly, $c_i^e$  denotes the end-point of the sub-sink $LSS(c_i)$.


Algorithm \ref{Algo:2} prioritizes the sub-sinks based on their components' last sub-sink's end-point positions and sub-sinks' start-point positions. The sub-sink whose component's last sub-sink's end-point appears first on $P$, gets the highest preference for data delivery. If two sub-sinks are on the same component, then the sub-sink whose start-point appears first has a higher priority than the other. The $MS$ receives data from a sub-sink in $L$, which has maximum priority and has data to deliver. If there is no data, then it is immediately removed from $L$.


Similar to the  proof of Theorem \ref{th:th2}, assume for the sake of contradiction that the $MS$ does not complete the data gathering process in minimum time. It implies that there is a sub-path of $[c_i^s, c_i^e]$  for component $c_i$, and the sub-path is under the communication disk of a sub-sink $ss_k \in c_i$, where the $MS$ moves without receiving data from any sub-sink or receives data from a sub-sink $ss_l$, whose priority $PR(ss_l) < PR(ss_k)$ and the $MS$ waits at $c_i^e$ for receiving data from component $c_i$. The sub-sink  $ss_l$  may belong to (i) same component $c_i$ as of $ss_k$, or (ii) in a different component $c_j$, i.e.  $c_j \ne c_i$.   

In case (i), where sub-sink  $ss_l \in c_i$, the $MS$ does not wait at $c_i^e$. This is because within the communication disk of $ss_k$, if the $MS$ receives data from sub-sink  $ss_l \in c_i$ with $PR(ss_l) < PR(ss_k)$, then all the sub-sinks in $c_i$ with priority $\ge PR(ss_k)$ do not have data to deliver. 

In case (ii), where sub-sink $ss_l \in c_j$ and $ c_j \ne c_i$, based on the first sub-sink's start-point and last sub-sink's end-point positions of a component, two components $c_i$ and $c_j$ have three different types of overlaps as shown in Figure \ref{fig:Proof_Th4}. All other types of overlaps are equivalent to one of them. We will show that extra delay at $c_i^e$ does not hold for any of these three types of overlaps.

\begin{figure}[!t]
\centering 
\includegraphics[width=8cm]{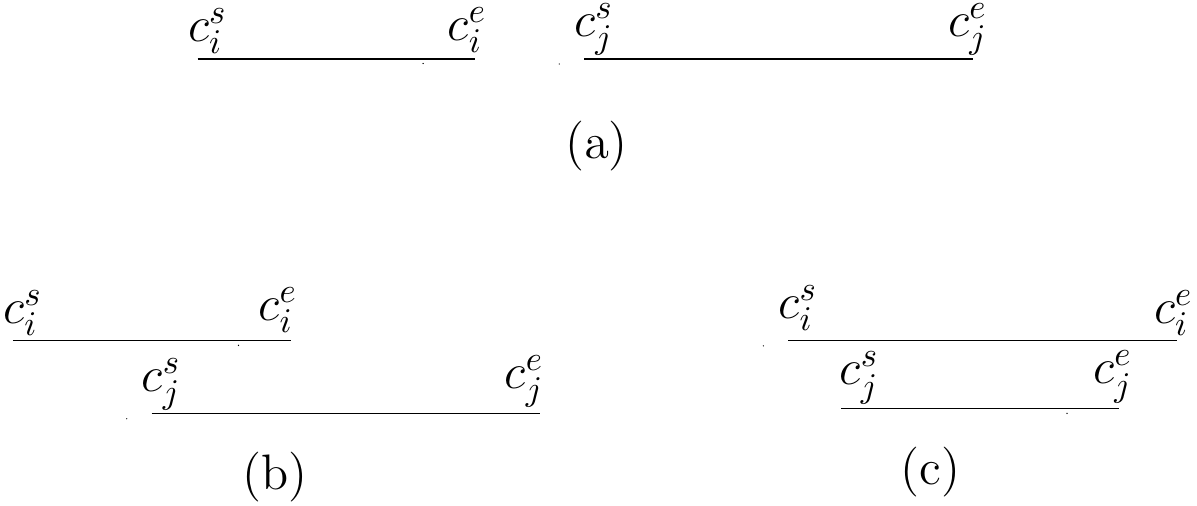}
\caption{Overlaps between components}
\label {fig:Proof_Th4}
\end{figure}

In Figure \ref{fig:Proof_Th4}(a) type overlap, two components are disjoint. Hence, the $MS$ does not receive data from $ss_l \in c_j$ within [$c_i^s, c_i^e$ ] and make an extra delay at $c_i^e$.
  

In Figure \ref{fig:Proof_Th4}(b) type overlap, the second component starts before the end of first component. In this case, the priority of any sub-sink in component $c_i$ is higher than any sub-sink in component $c_j$. Hence, the $MS$ receives data from $ss_l \in c_j$ within [$c_i^s, c_i^e$ ] only when there is no data in $ss_k$. If there is no data in $ss_k$, then all data from $c_i$ is already delivered to the $MS$ and the $MS$ does not wait at $c_i^e$.

In Figure \ref{fig:Proof_Th4}(c) type overlap, the priority of any sub-sink in component $c_j$ is higher than any sub-sink in component $c_i$. So, $PR(ss_l)$ can not be less than $PR(ss_k)$.


Therefore, in both case (i) and case (ii) our assumption does not hold and hence the theorem is proved.
\end{proof}

\FloatBarrier


\section{Experiment and Performance Analysis} \label{sec:experiment}

We evaluate the performance of our two proposed algorithms. We have used MATLAB for implementing our algorithms. In this section, we evaluate the performances of our proposed algorithms. We refer the algorithms for data gathering algorithm using speed controllable mobile-sink with known data availability (Algorithm \ref{Algo:1}) as {\bf VS-K-DA}. {\bf VS-UK-DA} refers to the case where data availabilities of the sub-sinks are unknown and optimized using Algorithm \ref{Algo:2}. Algorithm {\bf VS-UK-DA} is combined with network flow algorithm for data distribution and with Algorithm \ref{Algo:1} to find data gathering schedule. We compare the above two algorithms with a third algorithm {\bf FS-K-DA}, where data availabilities of the sub-sinks are known apriori as in {\bf VS-K-DA} and the mobile sink $MS$ is moving with its maximum speed as in \cite{Gao:2011} for collecting data from the sub-sinks.

\subsection{Simulation environment}

During simulation, the number of sensor nodes is varying for 100, 120, 140 and 160. The communication range of sensors is set to 75m. Sensors deployment region is a rectangular area of size 1000m x 400m. The rectangular region is vertically partitioned into four sub-regions of length 250m each. Within each sub-region of length 250m, sensors are randomly deployed within a vertical strip of [75m : 150m]. This is done to ensure that the random communication topology forms at least four connected components and there are gaps between the consecutive components. In the simulation, the $MS$ is moving along a horizontal path $P$ at the centre of the region (y=200m). The maximum speed $V$ of the $MS$ is set to 2 m/s. The $MS$ collects data from one sub-sink at a time, which is within the communication range. The data transfer rate between a sub-sink and the $MS$ is set to 2 Kbps. We assume that the sensors generate data randomly between 0 to 10 packets, and each packet is of size 1Kb. The far-away sensors send their sensed data to the sub-sinks through multi-hop forwarding. Data availabilities of the sub-sinks (for known apriori case) of problem 1 is determined using shortest path routing, where the sensors forward their data to its closest (hop-count) sub-sink.

Let $e_r$ and $e_t$ denote energy consumption for receiving and transmitting unit bit data. Let  $E_i$ represent the total energy consumption of a sensor $s_i$ for receiving $d_{r}^i$ bits, and transmitting $d_{t}^i $ bits. Therefore, $E_i$ can be written as : 

 \begin{equation}\label{eq:1}
  E_i = ( e_r * d_{r}^i + e_t *d_{t}^i  )
 \end{equation}

Total energy consumption of the network $E_{total}$ is calculated as the summation of energy consumption for forwarding data from the sensors to the $MS$ through their respective sub-sinks.

\begin{equation}\label{eq:2}
 E_{total} = \sum_{i=1}^{n} {E_i} 
\end{equation}

Let $DD(ss_i)$ denote data delivered by a sub-sink $ss_i$ to the $MS$. Hence, total energy consumption $E_{total}$ includes energy consumption for delivering data from the sub-sinks to the $MS$, which is $\sum_{i=1}^{m} {e_t *DD(ss_i) }$. This is because the sub-sinks $\mathbb{SS} \subseteq N$. Table \ref{tab:1} summarizes the simulation parameters.


\begin{table}[h]
\caption{Simulation Parameters} 
\label{tab:1}

\begin{tabular}{lll}

 \hline
{\bf Parameter}  & {\bf Value} \\

 \hline
Rectangular deployment area  &  1000m  $\times$  400m \\
No. of sensors  & 100, 120, 140, 160  \\
Maximum speed of $MS$ & 2 m/s \\
Communication range of sensor  & 75 m \\
Data transmission rate & 2 Kbps \\
$e_r$  & 2  $\mu$ Joule/bit  \\
$e_t$  & 3  $\mu$ Joule/bit  \\

 \hline
\end{tabular}
\end{table}

\subsection{Performance analysis of the proposed algorithms}

Figure \ref{fig:result1} shows the total data collected by the mobile-sink with respect to the number of sensors. From the figure, it is obvious that the amount of data collection is proportional to the number of sensors. Data collection in {\bf VS-K-DA} and  {\bf VS-UK-DA} are same because both the algorithms collect complete data from the network, whereas data collection in {\bf FS-K-DA} is lesser than the two proposed algorithms. The difference between fixed speed and variable speed data gathering increases as the number of sensors increases. This is because, in {\bf VS-K-DA} and  {\bf VS-UK-DA}, the complete data from the sub-sinks are collected by controlling the speed of the $MS$, whereas in {\bf FS-K-DA}, the $MS$ moves with its maximum speed, and the sub-sinks do not get enough time to deliver their data completely.

\begin{figure}[!h]
\centering 
\includegraphics[ width=8cm]{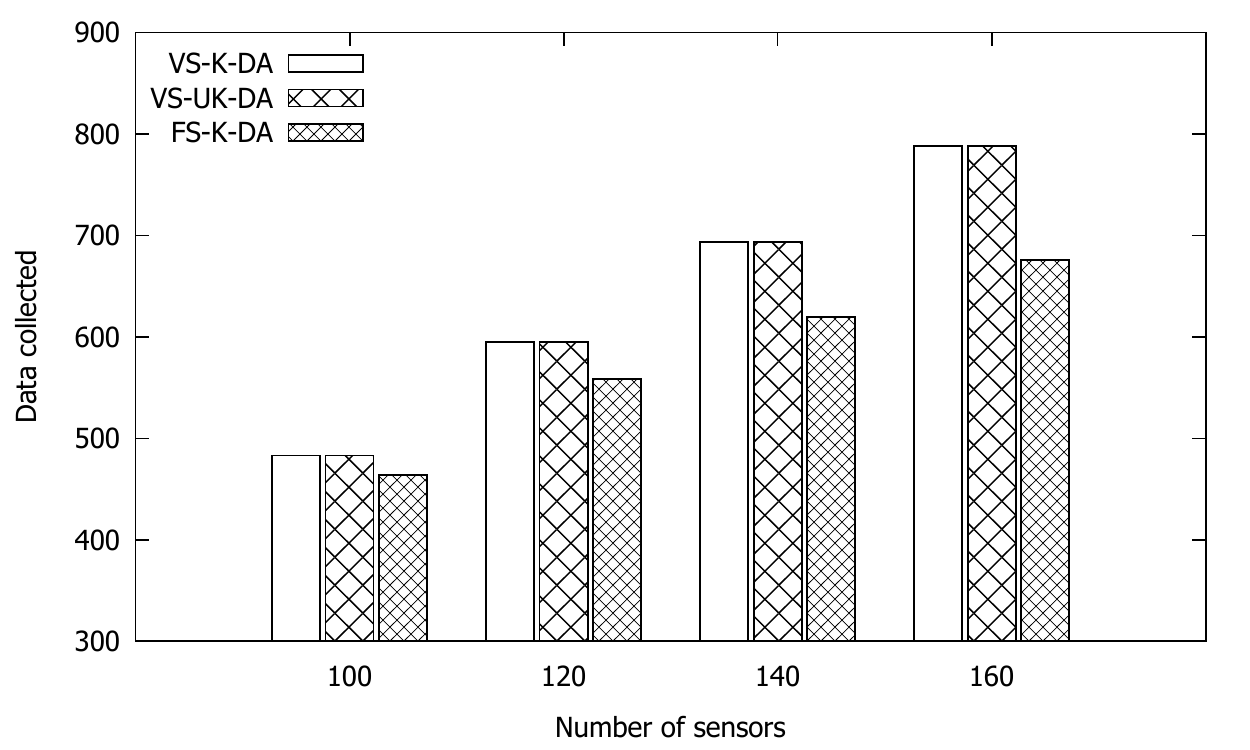}
\caption{Data collected (Kb) with respect to no of sensors}
\label{fig:result1}
\end{figure}


Figure \ref{fig:result2} shows the data gathering time with respect to the number of sensors. It shows that the data gathering time increases proportionally to the number of sensors for the two proposed algorithms {\bf VS-K-DA} and {\bf VS-UK-DA}. But data gathering time of {\bf FS-K-DA} is constant and it does not depend on the number of sensors. This is because in {\bf FS-K-DA} the $MS$ moves with its fixed maximum speed (2m/sec). Data gathering time in {\bf VS-UK-DA} is lesser than {\bf VS-K-DA}. The time difference between {\bf VS-K-DA} and  {\bf VS-UK-DA} increases proportional to the number of sensors present in the network. Because in {\bf VS-UK-DA}, sensors' data are forwarded to the sub-sinks to reduce the total data gathering time. In algorithm {\bf VS-UK-DA}, sometimes data are forwarded to the sub-sinks at a longer hop count distance. It increases the energy consumption of the network, which is reflected in Figure \ref{fig:result7}. 

\begin{figure}[!h]
\centering 
\includegraphics[ width=8cm]{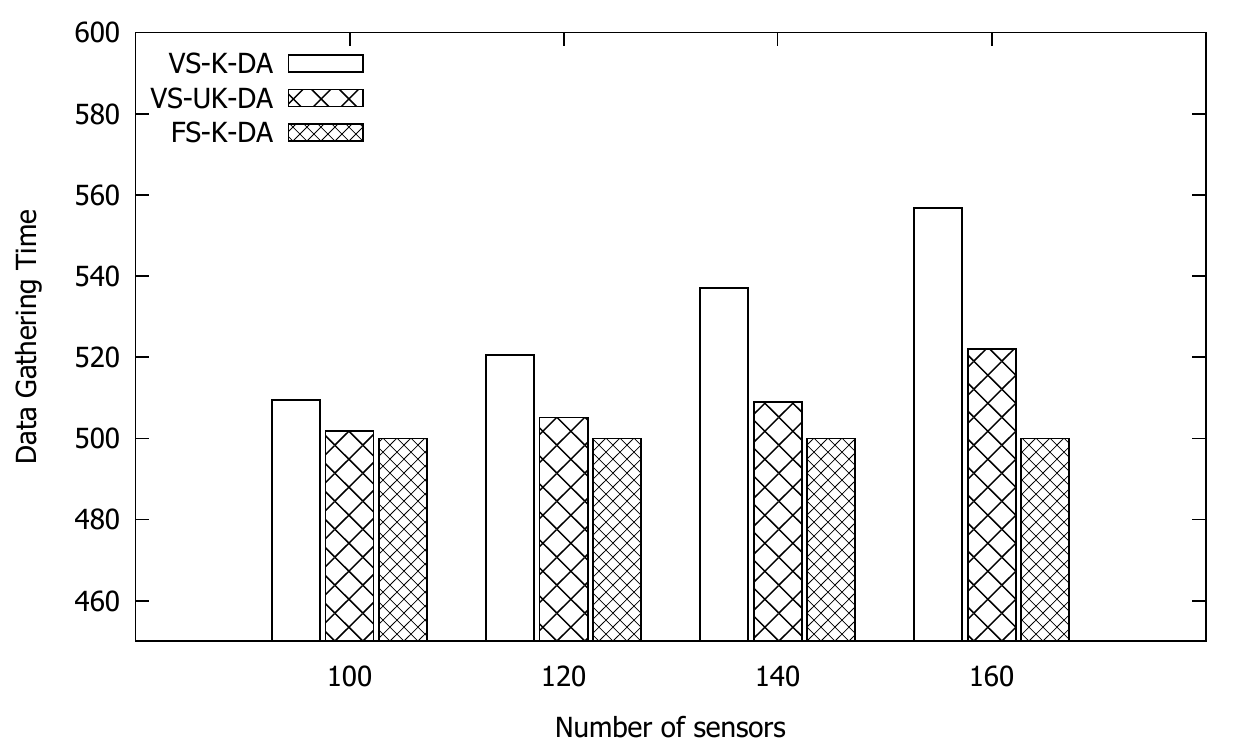}
\caption{Data gathering time (Sec) with respect to no of sensors }
\label{fig:result2}
\end{figure}


Figure \ref{fig:result3} shows the average speed of the $MS$ with respect to the number of sensors. For our two proposed algorithms, the average speed of the $MS$ decreases as the number of sensors increases. This is because as the number of sensors increase, more data are forwarded to the sub-sinks, and it increases the data transmission time from the individual sub-sink to the $MS$. The average speed of {\bf VS-UK-DA} is little higher than {\bf VS-K-DA}.

\begin{figure}[!h]
\centering 
\includegraphics[width=8cm]{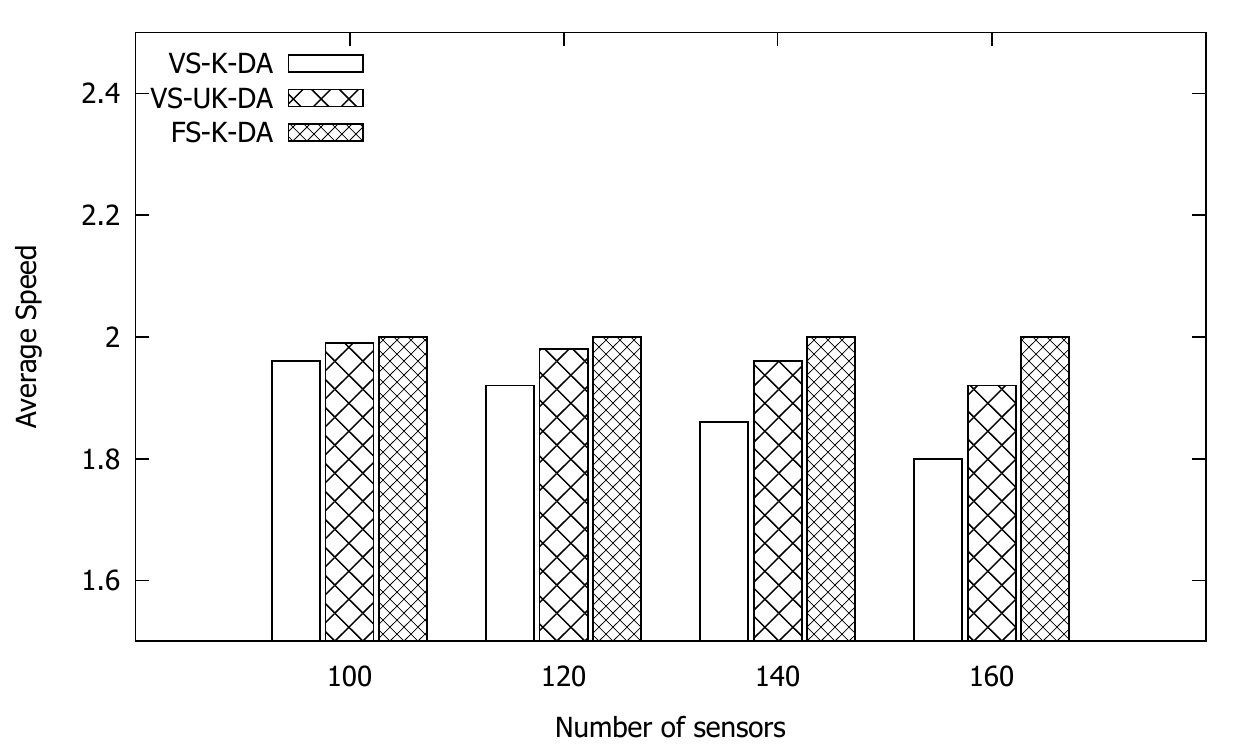}
\caption{Average speed (m/Sec) of the mobile sink }
\label{fig:result3}
\end{figure}

\begin{figure}[!h]
\centering 
\includegraphics[ width=8cm]{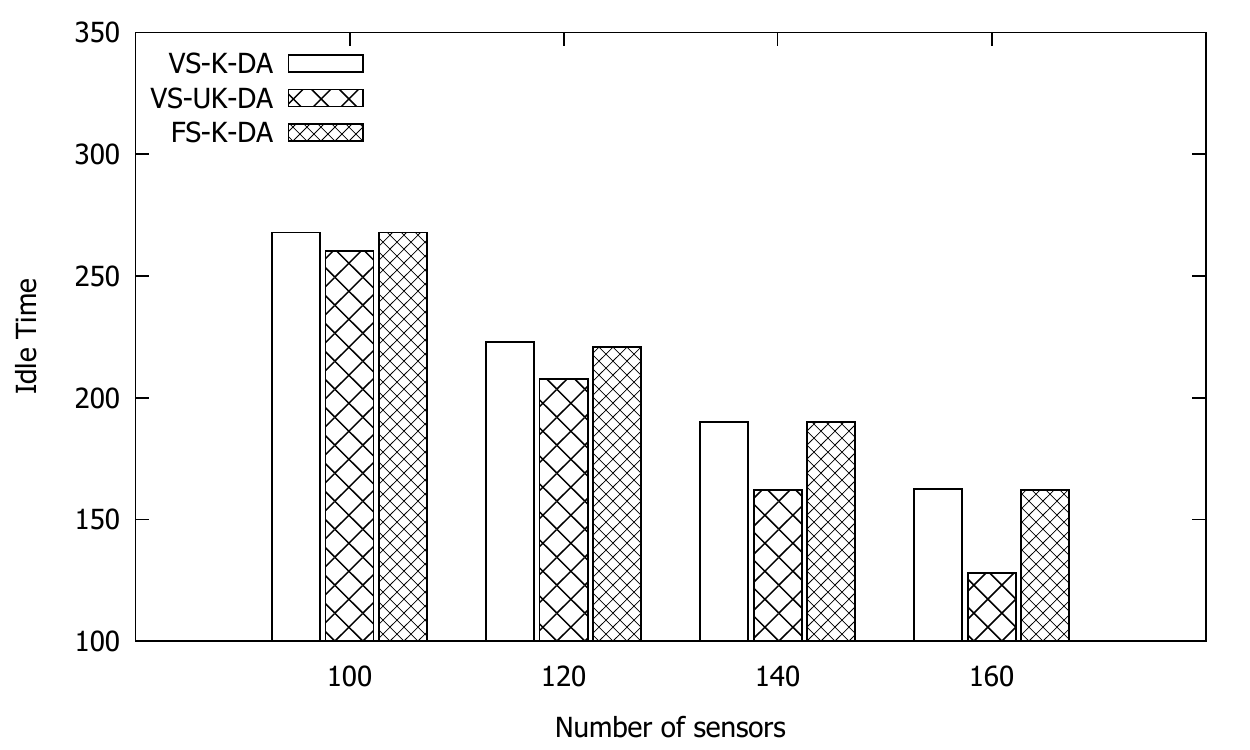}
\caption{Idle time (Sec) of the mobile sink }
\label{fig:result4}
\end{figure}

Idle period of the $MS$ denotes the time the $MS$ moves without receiving data from any sub-sink while moving on the path $P$. Figure \ref{fig:result4} shows the idle period of the $MS$. The idle period decreases as the number of sensors increases. This is because as the number of sensors increases, more sub-sinks are there and hence, the total data transfer time increases and the idle time decreases. Idle period of {\bf VS-UK-DA} is comparatively lower than the other, and the difference increases as the number of sensors increases.


Throughput measures the amount of data collected by the $MS$ per unit time. Figure \ref{fig:result5} shows the throughput of the network with respect to the number of sensors. As the number of sensors increases, the number of sub-sinks and the total data collection by the $MS$ are also increased and hence, improves the throughput of the network. From the result, it is observed that the throughput of {\bf VS-UK-DA} is comparatively higher than the other.

\begin{figure}[!h]
\centering 
\includegraphics[ width=8cm]{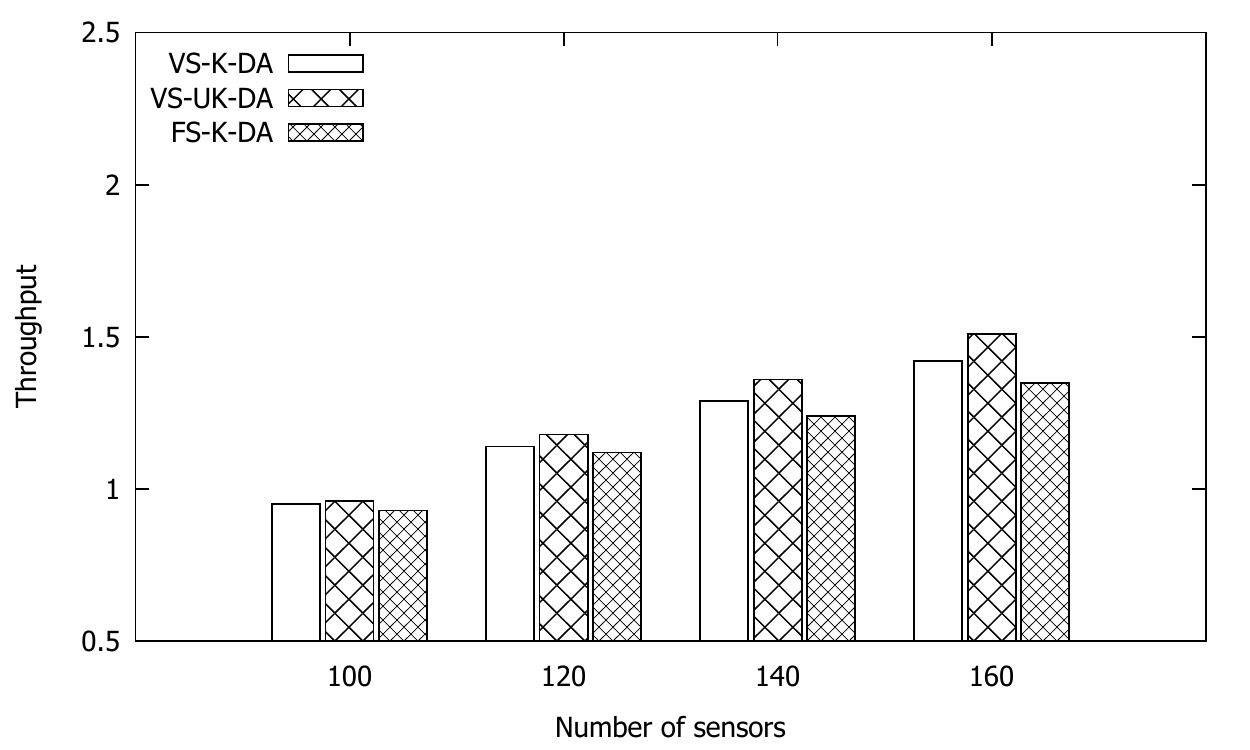}
\caption{Throughput with respect to no of sensors}
\label{fig:result5}
\end{figure}


\begin{figure}[!h]
\centering 
\includegraphics[ width=8cm]{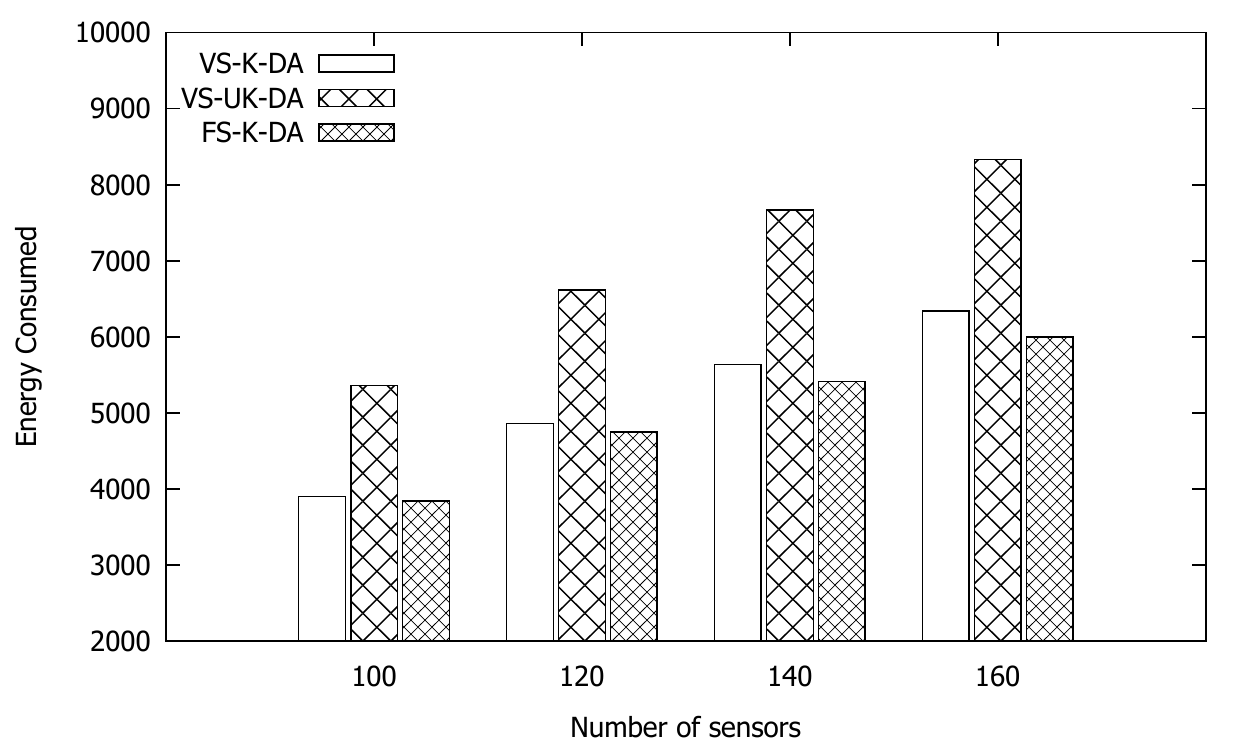}
\caption{Total Energy Consumption (m Joule) with respect to no of sensors}
\label{fig:result7}
\end{figure}

We evaluate the total energy consumption for forwarding data to the $MS$, but the energy consumption of the $MS$ is not considered. Figure \ref{fig:result7} shows the total energy consumption with respect to the number of sensors. As the number of sensors increases total data generated in the network increases proportionally and hence, total energy consumption increases proportionally. In both {\bf VS-K-DA} and  {\bf FS-K-DA} data gathering algorithms, the data generated in the sensors are transferred to the sub-sinks through the shortest path. But, in algorithm {\bf FS-K-DA}, complete data from the sub-sinks are not delivered to the $MS$ and hence, energy consumption is little lesser than {\bf VS-K-DA}. Whereas both {\bf VS-K-DA} and {\bf VS-UK-DA} algorithms deliver complete data to the $MS$, but algorithm  {\bf VS-UK-DA} forwards data to the sub-sinks to optimize total gathering time. Hence, sometimes sensors' data are forwarded to sub-sinks which are at a longer distance, which increases the total energy consumption of {\bf VS-UK-DA}.  

Finally, we study the performance of the algorithms, by varying the maximum speed limit $V$ of the $MS$ and evaluate the total data gathering time of the $MS$. Figure \ref{fig:result6} shows the total data gathering time for different speeds. As the maximum speed limit increases, the data gathering time also decreases. This is because as the speed increases, the idle period of the $MS$ decreases proportionally. Also, the time difference between {\bf VS-K-DA} and {\bf VS-UK-DA} decreases proportionally. For fixed speed data gathering {\bf FS-K-DA}, total data gathering time decreases linearly, which is reflected in the figure. 

\begin{figure}[!h]
\centering 
\includegraphics[ width=8cm]{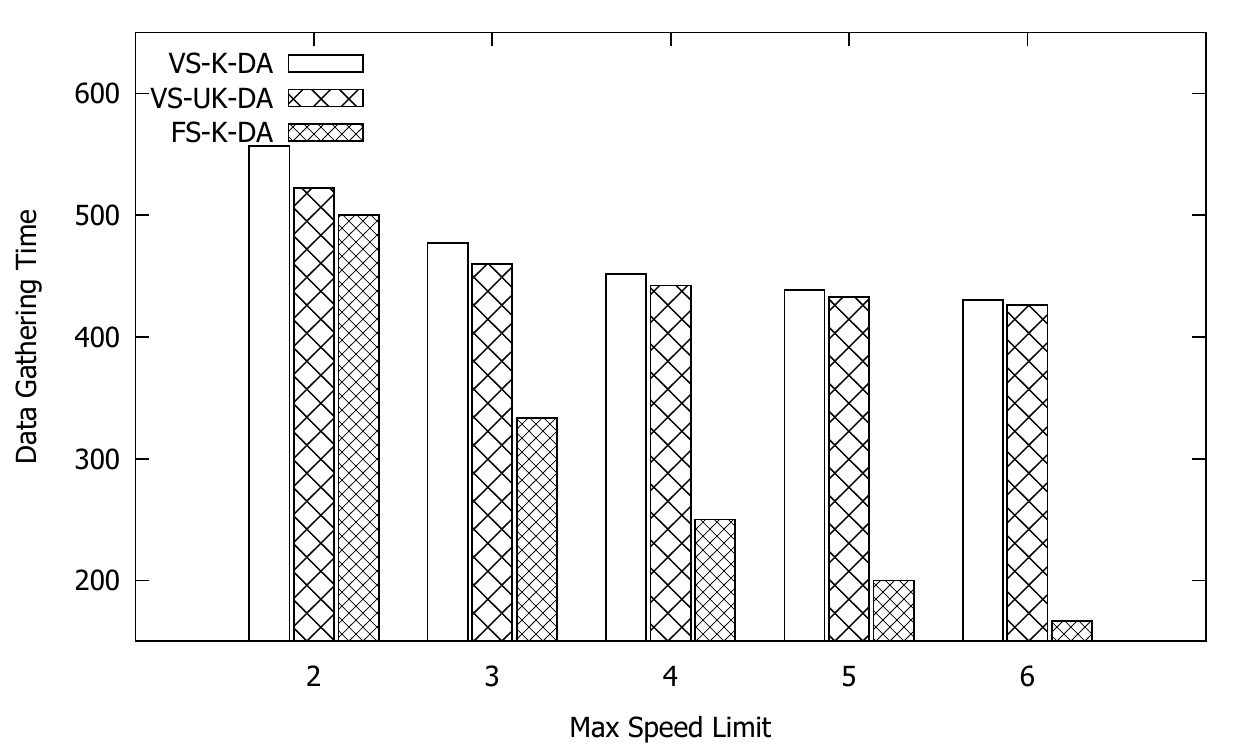}
\caption{ Data gathered time (Sec) with respect to maximum speed of the $MS$}
\label{fig:result6}
\end{figure}

\FloatBarrier
\section{Conclusion} \label{sec:conclusion}

In this article, we have studied two problems for the maximum data gathering using a mobile sink ($MS$) for time-sensitive applications. The $MS$ can adjust its movement speed while moving along a given path in the network. However, the speed of the $MS$ cannot go beyond a given maximum speed limit $V$. We have presented plane sweep based algorithms to find optimal data gathering schedule of the $MS$. In the first algorithm, the minimum time data gathering schedule of the mobile-sink is determined by controlling the data transmission schedule of the sub-sinks and speed of the $MS$, where the data availability values of the sub-sinks are known. The second algorithm improves the data gathering time and the throughput by optimizing the data availability values of the sub-sinks by controlling the data distribution from the sensors to the sub-sinks. It is observed from the experiment results that the data gathering time of Algorithm {\bf VS-UK-DA} is better than the Algorithm {\bf VS-K-DA}. But, energy consumption of {\bf VS-UK-DA} is higher than {\bf VS-K-DA} and {\bf FS-K-DA}. The results also show that both {\bf VS-K-DA} and {\bf VS-UK-DA} have better data gathering capability and throughput than {\bf FS-K-DA}. In future, we plan to find an optimal fixed speed of the $MS$ to improve the total data collection process. In addition, we will find an optimal path for the $MS$ to improve data collection for time-sensitive applications.


\section{Acknowledgments} \label{sec11}

This work is supported by the Science \& Engineering Research Board, DST, Govt. of India [Grant numbers: ECR/2016/001035]

\bibliographystyle{abbrv}
\bibliography{References}



\begin{IEEEbiography}[{\includegraphics[width=1in, height=1.25in, clip, keepaspectratio]{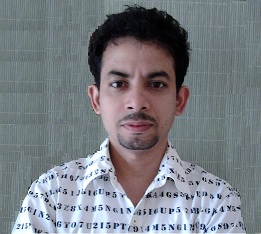}}]{Dr. Dinesh Dash} received Master of Technology in Computer Science and Engineering from University of Calcutta, India in 2004. From 2004 to 2007 he worked as a Lecturer under West Bengal University of Technology, India. He was awarded Ph.D. in 2013 from Indian Institute of Technology Kharagpur, India. His PhD research topics was on Coverage Problems in Wireless Sensor Network.  He worked as senior research associate from 2013 to 2014 at Infosys Limited, India. From 2013 to 2014 he worked as Assistant Professor at Tezpur University, Assam, India. Since Dec 2014, he is working as an Assistant Professor in the Dept of CSE, National Institute of Technology Patna, India. His current work focuses on sensor network coverage problem, data gathering problem, design of fault tolerant system in mobile AdHoc Network.
\end{IEEEbiography}

\end{document}